\documentclass[runningheads]{llncs}

\usepackage{graphicx}
\usepackage{textcomp}
\usepackage{xcolor}

\usepackage[ruled, vlined, nofillcomment, linesnumbered]{algorithm2e}
\usepackage{url}
\usepackage{amsmath}
\usepackage{graphicx}
\usepackage{subcaption}
\usepackage[english]{babel}
\usepackage{booktabs}
\usepackage{verbatim}
\usepackage{float}
\usepackage{footmisc}
\usepackage{xspace}
\usepackage{tikz}
\usepackage{pgfplots}
\usepackage{pgfplotstable}

\usepackage[numbers,sort&compress]{natbib}

\usepackage{etoolbox}

\SetKw{Output}{output}

\newcommand{\set}[1]{\left\{#1\right\}}
\newcommand{\pr}[1]{\left(#1\right)}
\newcommand{\fpr}[1]{\mathopen{}\left(#1\right)}
\newcommand{\spr}[1]{\left[#1\right]}

\newcommand{\abs}[1]{{\left|#1\right|}}

\newcommand{\np}{\textbf{NP}}

\newcommand{\funcdef}[3]{{#1}:{#2} \to {#3}}
\newcommand{\define}{\leftarrow}

\newcommand{\fm}[1]{{\mathcal{#1}}}

\DeclareRobustCommand{\dispfunc}[2]{%
	\ensuremath{%
		\ifthenelse{\equal{#2}{}}%
			{\mathit{#1}}%
			{\mathit{#1}\fpr{#2}}}}

\newcommand{\lhood}[1]{\dispfunc{\ell}{#1}}

\newcommand{\bigO}[1]{\dispfunc{\mathcal{O}}{#1}}

\newcommand{\duration}[1]{\dispfunc{\Delta}{#1}}
\newcommand{\cnt}[1]{\dispfunc{c}{#1}}

\newcommand{\dtname}[1]{\textsl{#1}}

\newcommand{\findgroups}{\textsc{FindGroups}\xspace}
\newcommand{\findsegments}{\textsc{FindSegments}\xspace}
\newcommand{\findlambda}{\textsc{UpdateLambda}\xspace}

\SetKw{Break}{break}

\pgfdeclarelayer{background}
\pgfdeclarelayer{foreground}
\pgfsetlayers{background,main,foreground}

\definecolor{yafaxiscolor}{rgb}{0.3, 0.3, 0.3}

\definecolor{yafcolor1}{rgb}{0.4, 0.165, 0.553}
\definecolor{yafcolor2}{rgb}{0.949, 0.482, 0.216}
\definecolor{yafcolor3}{rgb}{0.47, 0.549, 0.306}
\definecolor{yafcolor4}{rgb}{0.925, 0.165, 0.224}
\definecolor{yafcolor5}{rgb}{0.141, 0.345, 0.643}
\definecolor{yafcolor6}{rgb}{0.965, 0.933, 0.267}
\definecolor{yafcolor7}{rgb}{0.627, 0.118, 0.165}
\definecolor{yafcolor8}{rgb}{0.878, 0.475, 0.686}

\newlength{\yafaxispad}
\newlength{\yaftlpad}
\newlength{\yaflabelpad}
\newlength{\yafaxiswidth}
\newlength{\yafticklen}
\makeatletter
\def\pgfplots@drawtickgridlines@INSTALLCLIP@onorientedsurf#1{}
\makeatother

\newcommand{\yafdrawaxis}[4]{
	\pgfplotstransformcoordinatex{#1}\let\xmincoord=\pgfmathresult 
	\pgfplotstransformcoordinatex{#2}\let\xmaxcoord=\pgfmathresult 
	\pgfplotstransformcoordinatey{#3}\let\ymincoord=\pgfmathresult 
	\pgfplotstransformcoordinatey{#4}\let\ymaxcoord=\pgfmathresult 
	\pgfsetlinewidth{\yafaxiswidth} 
	\pgfsetcolor{yafaxiscolor}
	\pgfpathmoveto{\pgfpointadd{\pgfpointadd{\pgfplotspointrelaxisxy{0}{0}}{\pgfqpointxy{\xmincoord}{0}}}{\pgfqpoint{-0.5\yafaxiswidth}{\yafaxispad}}}
	\pgfpathlineto{\pgfpointadd{\pgfpointadd{\pgfplotspointrelaxisxy{0}{0}}{\pgfqpointxy{\xmaxcoord}{0}}}{\pgfqpoint{0.5\yafaxiswidth}{\yafaxispad}}}
	\pgfpathmoveto{\pgfpointadd{\pgfpointadd{\pgfplotspointrelaxisxy{0}{0}}{\pgfqpointxy{0}{\ymincoord}}}{\pgfqpoint{\yafaxispad}{-0.5\yafaxiswidth}}}
	\pgfpathlineto{\pgfpointadd{\pgfpointadd{\pgfplotspointrelaxisxy{0}{0}}{\pgfqpointxy{0}{\ymaxcoord}}}{\pgfqpoint{\yafaxispad}{0.5\yafaxiswidth}}}
	\pgfusepath{stroke}
}

\pgfplotscreateplotcyclelist{yaf}{%
{yafcolor5,mark options={scale=0.75},mark=o}, 
{yafcolor2,mark options={scale=0.75},mark=square},
{yafcolor3,mark options={scale=0.75},mark=triangle},
{yafcolor4,mark options={scale=0.75},mark=o},
{yafcolor1,mark options={scale=0.75},mark=o},
{yafcolor6,mark options={scale=0.75},mark=o},
{yafcolor7,mark options={scale=0.75},mark=o},
{yafcolor8,mark options={scale=0.75},mark=o}} 

\pgfkeys{/pgf/number format/.cd,1000 sep={\,}}

\pgfplotsset{axis y line=left, axis x line=bottom,
	tick align=outside,
	tickwidth=\yafticklen,
	clip = false,
    x axis line style= {-, line width = 0pt, color=black!0},
    y axis line style= {-, line width = 0pt, color=black!0},
    x tick style= {line width = \yafaxiswidth, color=yafaxiscolor, yshift = \yafaxispad},
    y tick style= {line width = \yafaxiswidth, color=yafaxiscolor, xshift = \yafaxispad},
    x tick label style = {font=\small, yshift = \yaftlpad, inner xsep = 0pt},
    y tick label style = {font=\small, xshift = \yaftlpad},
    every axis y label/.style = {at = {(ticklabel cs:0.5)}, rotate=90, anchor=center, font=\small, yshift = -\yaflabelpad, inner sep = 0pt},
    every axis x label/.style = {at = {(ticklabel cs:0.5)}, anchor=center, font=\small, yshift = \yaflabelpad},
    x tick label style = {font=\small, yshift = 1pt},
    grid = major,
    major grid style  = {dash pattern = on 1pt off 3 pt},
	every axis plot post/.append style= {line width=\yafaxiswidth} ,
	legend cell align = left,
	legend style = {inner sep = 1pt, cells = {font=\scriptsize}},
	legend image code/.code={%
		\draw[mark repeat=2,mark phase=2,#1] 
		plot coordinates { (0cm,0cm) (0.15cm,0cm) (0.3cm,0cm) };%
	} 
}

\newcommand{\pgfprintduration}[1]{%
	\ifthenelse{\equal{#1}{}}{---}{%
	\pgfmathsetmacro{\minutes}{floor(#1 / 60)}%
	\pgfmathsetmacro{\seconds}{#1 - 60*\minutes}%
	\pgfmathifthenelse{\minutes > 0}{"\pgfmathprintnumber{\minutes}m \pgfmathprintnumber[fixed,precision=0]{\seconds}s"}{"\pgfmathprintnumber{\seconds}s"}\pgfmathresult}}

\begin{document}

\title{Recurrent segmentation meets block models in temporal networks}

\author{Chamalee Wickrama Arachchi\inst{1} \and Nikolaj Tatti\inst{1}}
\authorrunning{C. W. Arachchi and N. Tatti}
%
\institute{HIIT, University of Helsinki, Finland, \email{firstname.lastname@helsinki.fi}}


\maketitle
\begin{abstract}
A popular approach to model interactions is to represent them as a network with
nodes being the agents and the interactions being the edges. Interactions are
often timestamped, which leads to having timestamped edges.
Many real-world temporal networks have a recurrent or possibly cyclic behaviour. For
example, social network activity may be heightened during certain hours of day.
In this paper,
our main interest is to model recurrent activity in such temporal networks. As a starting point
we use stochastic block model, a popular choice for modelling static networks, where nodes
are split into $R$ groups.
We extend this model to temporal networks by modelling the edges with a Poisson process.
We make the parameters of the process dependent on time by segmenting the time line into $K$
segments. To enforce the recurring activity we require that only $H < K$ different set of parameters
can be used, that is, several, not necessarily consecutive, segments must share their parameters.
We prove that the searching for optimal blocks and segmentation is an \np-hard problem. Consequently,
we split the problem into 3 subproblems where we optimize blocks, model parameters, and segmentation
in turn while keeping the remaining structures fixed. We propose an iterative algorithm that requires
$\bigO{KHm + Rn + R^2H}$ time per iteration, where $n$ and $m$ are the number of nodes and edges in the network.
We demonstrate experimentally that the number of required iterations is typically low, the algorithm is able 
to discover the ground truth from synthetic datasets, and show that certain real-world networks
exhibit recurrent behaviour as the likelihood does not deteriorate when $H$ is lowered.
\end{abstract}

\section{Introduction}

A popular approach to model interactions between set of agents is to represent
them as a network with nodes being the agents and the interactions being the
edges. Naturally, many interactions in real-world datasets have a timestamp, in which
case the edges in networks also have timestamps. Consequently, developing
methdology for temporal networks has gained attention in data
mining literature.

Many temporal phenomena have recurrent or possibly cyclic behaviour.  For
example, social network activity may be heightened during certain hours of day.
Our main interest is to model recurrent activity in temporal networks. As a starting point
we use stochastic block model, a popular choice for modelling static networks. We can immediately
extend this model to temporal networks, for example, by modelling the edges with a Poisson process.  
Furthermore, \citet{corneli2018multiple} modelled the network by also segmenting the timeline and modelled each segment
with a separate Poisson process. 

To model the recurrent activity we can either model it explicitly, for example,
by modelling explicitly cyclic activity, or we can use more flexible approach
where we look for segmentation but restrict the number of distinct parameters.
Such notion was proposed by~\citet{gionis2003finding} in the context of segmenting
sequences of real valued vectors.

In this paper we extend the model proposed by \citet{corneli2018multiple} using the ideas
proposed by \citet{gionis2003finding}.
More formally, we consider the following problem: given a temporal graph with $n$ nodes and $m$
edges, we are looking to partition the nodes into $R$ groups and segment the timeline into $K$
segments that are grouped into $H$ levels. Note that a single level may contain non-consecutive  
segments.
An edge $e = (u, v)$ is then modelled with a Poisson process with a parameter $\lambda_{ijh}$,
where $i$ and $j$ are the groups of $u$ and $v$, and $h$ is the level of the segment containing $e$.

To obtain good solutions we rely on an iterative method by splitting
the problem into three subproblems: ($i$) optimize blocks while keeping the remaining parameters fixed,
($ii$) optimize model parameters $\Lambda$ while keeping the blocks and the segmentation fixed,
($iii$) optimize the segmentation while keeping the remaining parameters fixed.
We approach the first subproblem by iteratively optimizing block assignment of
each node while maintaining the remaining nodes fixed. We show that such single round 
can be done in $\bigO{m + Rn + R^2H + K}$ time, where $n$ is the number of nodes and $m$ is the number of edges.
Fortunately, the second subproblem is trivial since there is an analytic solution
for optimal parameters, and we can obtain the solution in $\bigO{m + R^2H + K}$ time.
Finally, we show that we can find the optimal segmentation with a dynamic
program.  Using a stock dynamic program leads to a computatonal complexity of $\bigO{m^2KH}$.
Fortunately, we show that we can speed up the computation by using a SMAWK
algorithm~\citep{aggarwal1987geometric}, leading to a computational complexity
of  $\bigO{mKH + HR^2}$.

In summary, we extend a model by~\citet{corneli2018multiple} to have recurring
segments. We prove that the main problem is \np-hard as well as
several related optimization problems where we fix a subset of parameters.
Navigating around these \np-hard problems we propose an iterative algorithm
where a single iteration requires $\bigO{KHm + Rn + R^2H}$ time, a linear time in edges and nodes.

The rest of the paper is organized as follows.
First we introduce preliminary notation, the model, and the optimization problem
in Section~\ref{sec:prel}. We then proceed to describe the iterative algorithm
in Section~\ref{sec:algorithm}. We present the related work in Section~\ref{sec:related}.
Finally, we present our experiments in Section~\ref{sec:exp} and conclude the paper
with discussion in Section~\ref{sec:conclusions}.
The proofs are provided in Appendix\footnote{\url{XXX}}.

\section{Preliminary notation and problem definition}\label{sec:prel}

Assume a \emph{temporal graph} $G = (V, E)$, where $V$ is a set of nodes
and $E$ is a set of edges, where each edge is tuple $(u, v, t$)  with $u,v \in V$
and $t$ being the timestamp. We will use $n = \abs{V}$ to denote the number of nodes
and $m = \abs{E}$ the number of edges. For simplicity, we assume that we
do not have self-loops, though the models can be adjusted for such case.
We write $t(e)$ to mean the timestamp of the edge $e$. We also write $N(u)$
to denote all the edges adjacent to a node $u \in V$.

Perhaps the simplest way to model a graph (with no temporal information) is
with Erdos-Renyi model, where each edge is sampled independently from a Bernoulli
probability parameterized with $q$.
Let us consider two natural extensions of this model. The first extension is
a block model, where nodes are divided into $k$ blocks, and an edge $(u, v)$ are modelled
with a Bernoulli probability parameterized with $q_{ij}$, where $i$ is the block of $u$ and
$j$ is the block of $v$. Given a graph, the optimization problem is to cluster nodes into blocks
so that the likelihood of the model is optimized. For the sake of variability we will use
the words \emph{block} and \emph{group} interchangeably.

A convenient way of modelling events in temporal data is using Poisson process:
Assume that you have observed $c$ events with timestamps $t_1, \ldots, t_c$ in a time interval $T$ of length $\Delta$.
The log-likelihood of observing these events at these exact times is equal to $c\log \lambda - \lambda \Delta$,
where $\lambda$ is a model parameter. Note that the log-likelihood does not depend on the individual timestamps. 

If we were to extend the block model to temporal networks, the log-likelihood of
$c$ edges occurring between the nodes $u$ and $v$ in a time interval is equal to
	$c\log \lambda_{ij} - \lambda_{ij} \Delta$,
where $\lambda_{ij}$ is the Poisson process parameter and $i$ is the block of $u$ and
$j$ is the block of $v$. Note that $\lambda_{ij}$ does not depend on the time, so 
discovering optimal blocks is very similar to discovering blocks in a static model.

A natural extension of this model, proposed by~\citet{corneli2018multiple}, is to make the parameters depend on time.
Here, we partition the model into $k$ segments and assign different set of $\lambda$s to each segment.

More formally, we define a \emph{time interval} $T$ to be a continuous
interval either containing the starting point $T = [t_1, t_2]$ or excluding the
starting point $T = (t_1, t_2]$. In both cases, we define the duration as
$\duration{T} = t_2 - t_1$.

Given a time interval $T$, let us define 
\[
	\cnt{u, v, T} = \abs{\set{e = (u, v, t) \in E \mid t \in T}}
\]
to be the number of edges between $u$ and $v$ in $T$.

The log-likelihood of Poisson model
for nodes $u$,
$v$ and a time interval $T$ is 
\[
	\lhood{u, v, T, \lambda} = \cnt{u, v, T} \log \lambda - \lambda \duration{T}\quad.
\]
We extend the log-likelihood between the two sets of nodes $U$ and $W$, by writing
\[
	\lhood{U, W, T, \lambda} = \sum_{u, w \in  U \times W} \lhood{u, w, T, \lambda},
\]
where $U \times W$ is a set of all node pairs $\set{u, w}$ with $u \in U$ and $w \in W$ and $u \neq v$.
We consider $\set{u, w}$ and $\set{w, u}$ the same, so only one of these pairs is visited.

Given a time interval $D = [a, b]$,
a $K$-segmentation $\mathcal{T} = T_1, \ldots, T_K$ is a sequence of $K$ time
intervals, such that $T_1 = [a, t_1], T_2 = (t_1, t_2], \ldots T_i = (t_{i - 1}, t_{i}], \ldots$, and
$T_{K} = (t_{K - 1}, b]$.
For notational simplicity, we require that the boundaries $t_i$
must collide with the timestamps of individual edges. We also assume that $D$ covers the edges.
If $D$ is not specified, then it is set to be the smallest interval covering the edges.

Given a $K$-segmentation, a partition of nodes $\mathcal{P} = P_1, \ldots, P_R$
into $R$ groups, and a set of $KR(R + 1)/2$ parameters $\Lambda = \set{\lambda_{ijk}}$\footnote{For notational simplicity we will equate $\lambda_{ijh}$ and $\lambda_{jih}$.},
the log-likelihood is equal to
\[
	\lhood{\mathcal{P}, \mathcal{T}, \Lambda} = \sum_{i = 1}^R \sum_{j = i}^R \sum_{k = 1}^K \lhood{P_i, P_j, T_k, \lambda_{ijk}}\quad.
\]

This leads immediately to the problem considered by~\citet{corneli2018multiple}.
\begin{problem}[$(K, R)$ model]
Given a temporal graph $G$, a time interval $D$, integers $R$ and $K$, find a node partition with $R$ groups, 
a $K$-segmentation, and a set of parameters $\Lambda$ so that $\lhood{\mathcal{P}, \mathcal{T}, \Lambda}$
is maximized.
\end{problem}

We should point out that for fixed $\mathcal{P}$ and $\mathcal{T}$, the optimal $\Lambda$
is equal to
\[
	\lambda_{ijk} = \frac{\cnt{P_i, P_j, T_k}}{\abs{P_i \times P_j} \duration{T_k}}\quad.
\]

In this paper we consider an extension of $(K, R)$ model. Many temporal network exhibit cyclic or repeating behaviour.
Here, we allow network to have $K$ segments but we also limit the number of distinct parameters to be at most $H \leq K$.
In other words, we are forcing that certain segments \emph{share} their parameters. We do not know beforehand which segments should
share the parameters.

We can express this constraint more formally by introducing a mapping $\funcdef{g}{\spr{K}}{\spr{H}}$ that
maps a segment index to its matching parameters. We can now define the likelihood as follows:
given a $K$-segmentation, a partition of nodes $\mathcal{P} = P_1, \ldots, P_R$
into $R$ groups, a mapping $\funcdef{g}{\spr{K}}{\spr{H}}$,
and a set of $HR(R + 1)/2$ parameters $\Lambda = \set{\lambda_{ijh}}$,
the log-likelihood is equal to
\[
	\lhood{\mathcal{P}, \mathcal{T}, g, \Lambda} = \sum_{i = 1}^R \sum_{j = i}^R \sum_{k = 1}^K \lhood{P_i, P_j, T_k, \lambda_{ijg(k)}}\quad.
\]
We will refer to $g$ as \emph{level mapping}.

This leads to the following optimization problem.
\begin{problem}[$(K, H, R)$ model]
\label{prob:khr}
Given a temporal graph $G$, a time interval $D$, integers $R$, $H$, and $K$, find a node partition with $R$ groups, 
a $K$-segmentation, a level mapping $\funcdef{g}{\spr{K}}{\spr{H}}$,
and parameters $\Lambda$ maximizing $\lhood{\mathcal{P}, \mathcal{T}, g, \Lambda}$.
\end{problem}

\section{Fast algorithm for obtaining good model}\label{sec:algorithm}

In this section we will introduce an iterative, fast approach for obtaining a good model.
The computational complexity of one iteration is $\bigO{KHm + Rn + R^2H}$,
which is linear in both the nodes and edges.

\subsection{Iterative approach}

Unfortunately, finding optimal solution for our problem is \np-hard.
\begin{proposition}
\label{prop:nphard1}
Problem~\ref{prob:khr}
is \np-hard, even for $H = K = 1$ and $R = 2$.
\end{proposition}
Consequently, we resort to a natural heuristic approach,
where we optimize certain parameters while keeping the remaining parameters fixed.

\begin{algorithm}[t!]
\caption{Main loop of the algorithm}
\label{alg:main}
$\mathcal{P} \define$ random groups; \  $\Lambda \define$ random values\;
$\mathcal{T}, g \define \findsegments(\mathcal{P}, \Lambda)$\;
$\Lambda \define \findlambda(\mathcal{P}, \mathcal{T}, g)$\;
\While {convergence} {
	$\mathcal{P} \define \findgroups(\mathcal{P}, \Lambda, \mathcal{T}, g)$\;
	$\Lambda \define \findlambda(\mathcal{P}, \mathcal{T}, g)$\;
	$\mathcal{T}, g \define \findsegments(\mathcal{P}, \Lambda)$\;
	$\Lambda \define \findlambda(\mathcal{P}, \mathcal{T}, g)$\;
}
\end{algorithm}

We split the original problem into 3 subproblems as shown in Algorithm~\ref{alg:main}.
First, we find good groups, then update $\Lambda$, and then optimize segmentation, followed
by yet another update of $\Lambda$.

When initializing, we select groups $\mathcal{P}$ and parameters $\Lambda$ randomly, 
then proceed to find optimal segmentation, followed by opimizing $\Lambda$.

Next we will explain each step in details.

\subsection{Finding groups}

Our first step is to update groups $\mathcal{P}$ while maintaining
the remaining parameters fixed. Unfortunately, finding the optimal solution for
this problem is \np-hard.

\begin{proposition}
\label{prop:nphard2}
Finding optimal partition $\mathcal{P}$ for fixed $\Lambda$, $\mathcal{T}$ and $g$
is \np-hard, even for $H = K = 1$ and $R = 2$.
\end{proposition}

Due to the previous proposition, we perform a simple greedy optimization where
each node is individually reassigned to the optimal group while maintaining 
the remaining nodes fixed.

We should point out that there are more sophisticated approaches, for example
based on SDP relaxations, see a survey by~\citet{abbe2017community}. However, we resort to a simple
greedy optimization due to its speed.

A naive implementation of computing the log-likelihood gain for a single node may require
$\Theta(m)$ steps, which would lead in $\Theta(nm)$ time as we need to test every node.
Luckily,
we can speed-up the computation using the following straightforward proposition.

\begin{proposition}
\label{prop:gain}
Let $\mathcal{P}$ be the partition of nodes, $\Lambda$ set of parameters, and
$\mathcal{T}$ and $g$ the segmentation and the level mapping.
Let $\mathcal{S}_h = \set{T_k \in \mathcal{T} \mid h = g(k)}$ be the segments
using the $h$th level.

Let $u$ be a node, and let $P_b$ be the set such that $u \in P_b$.
Select $P_a$, and let $\mathcal{P}'$ be the partition where $u$ has been moved
from $P_b$ to $P_a$. Then
\[
	\lhood{\mathcal{P}', \mathcal{T}, g, \Lambda} - \lhood{\mathcal{P}, \mathcal{T}, g, \Lambda} 
	=  Z + \sum_{h = 1}^H \lambda_{bah} t_h + \sum_{j = 1}^R c_{jh} \log \lambda_{ajh} - \abs{P_j}\lambda_{ajh} t_h,
\]
where $Z$ is a constant, not depending on $a$, $t_h = \duration{\mathcal{S}_h}$ is the total duration of the segments using the $h$th level
and 
	$c_{jh} = \cnt{u, P_j, \mathcal{S}_h}$,
is the number of edges between $u$ and $P_j$ in the segments using the $h$th level.
\end{proposition}

The proposition leads to the pseudo-code given in Algorithm~\ref{alg:groups}.
The algorithm computes an array $c$ and then uses Proposition~\ref{prop:gain} to 
compute the gain for each swap, and consequently to find the optimal gain.

\begin{algorithm}[t!]

\caption{Algorithm $\findgroups(\mathcal{P}, \Lambda)$ for finding groups for a fixed segmentation $\mathcal{T}$, $g$ and parameters $\Lambda$}
\label{alg:groups}

$p(v) \define $ group index of $v$\;
$s(e) \define $ segment index of $e$\;

$d[h] \define \sum_{g(k) = h} \duration{T_k}$\;

\ForEach{$v \in V$} {
	$b \define p(v)$\;
	$c[j, h] \define$ array $c_{jh}$ as defined in Proposition~\ref{prop:gain}\;
	\ForEach{$a = 1, \ldots R$} {
		$x[a] \define \sum_{h = 1}^H \lambda_{bah}d[h] + \sum_{j = 1}^R c[j, h] \log \lambda_{ajh} - \abs{P_j}\lambda_{ajh} d[h] $ \label{alg:gain}\;
	}
	$p(v) \define \arg \max_a x[a]$ (update $\mathcal{P}$ also)\;
}
\Return $\mathcal{P}$\;
\end{algorithm}

Computing the array requirs iterating over the adjacent edges, leading to
$\bigO{\abs{N(v)}}$ time, and computing the gains requires $\bigO{R^2H}$ time.
Consequently, the computational complexity for \findgroups is $\bigO{m + R^2Hn
+ K}$.

The running time can be further optimized by modifying Line~\ref{alg:gain}. There are at most
$2m$ non-zero $c[i, j]$ entries (across all $v \in V$), consequently we can speed up the computation
of a second term
by ignoring the zero entries in $c[i, j]$. In addition, for each $a$, the remaining terms
\[
	\sum_{h = 1}^H \lambda_{bah}d[h] + \sum_{j = 1}^R \abs{P_j}\lambda_{ajh} d[h]
\]
can be precomputed in $\bigO{RH}$ time and maintained in $\bigO{1}$ time.
This leads to a running time of $\bigO{m + Rn + R^2H + K}$.

\subsection{Updating Poisson process parameters}

Our next step is to update $\Lambda$ while maintaining the rest of the parameters fixed. This refers to \findlambda in Algorithm~\ref{alg:main}.
Fortunately, this step is straightforward as the optimal parameters are equal to
\[
	\lambda_{ijh} = \frac{\cnt{P_i, P_j, \mathcal{S}_h}}{\abs{P_i \times P_j} \duration{\mathcal{S}_h}},
\]
where $\mathcal{S}_h = \set{T_k \in \mathcal{T} \mid h = g(k)}$ are the segments using the $h$th level.
Updating the parameters requires $\bigO{m + R^2H + K}$ time.

In practice, we would like to avoid having $\lambda = 0$ as this forbids any edges
occurring in the segment, and we may get stuck in a local maximum. We approach this by shifting $\lambda$
slightly by using 
\[
	\lambda_{ijh} = \frac{\cnt{P_i, P_j, \mathcal{S}_h} + \theta}{\abs{P_i \times P_j} \duration{\mathcal{S}_h} + \eta},
\]
where $\theta$ and $\eta$ are user parameters. 

\subsection{Finding segmentation}

Our final step is to update the segmentation $\mathcal{T}$ and the level
mapping $g$, while keeping $\Lambda$ and $\mathcal{P}$ fixed. Luckily, we can
solve this subproblem in linear time.

Note that we need to keep $\Lambda$ fixed, as otherwise the problem is \np-hard.
\begin{proposition}
\label{prop:nphard3}
Finding optimal $\Lambda$, $\mathcal{T}$ and $g$ for fixed $\mathcal{P}$
is \np-hard.
\end{proposition}

On the other hand, if we fix $\Lambda$, then we can solve the optimization problem
with a dynamic program. To be more specific, assume that the edges in $E$ are ordered, and write $o[e, k]$ to be the log-likelihood
of $k$-segmentation covering the edges prior and including $e$. Given two edges $s, e \in E$, let $y(s, e; h)$
be the log-likelihood of a segment $(t(s), t(e)]$ using the $h$th level of parameters, $\lambda_{\cdot\cdot h}$.
If $s$ occurs after $e$ we set $y$ to be $-\infty$.
Then the identity
\[
	o[e, k] = \max_h \max_s y(s, e; h) + o[s, k - 1]
\]
leads to a dynamic program.

Using an off-the-shelf approach by~\citet{bellman:61:on} leads to a computational complexity of $\bigO{m^2 K H}$, assuming that
we can evaluate $y(s, e; h)$ in constant time. 

However, we can speed-up the dynamic program by using a SMAWK algorithm~\citep{aggarwal1987geometric}. Given a function $x(i, j)$,
where $i, j = 1, \ldots, m$, SMAWK computes $z(j) = \arg \max_i x(i, j)$ in $\bigO{m}$ time, under two assumptions.
The first assumption is that we can evaluate $x$ in constant time. The second assumption is that $x$ is \emph{totally monotone}.
We say that $x$ is totally monotone,
if $x(i_2, j_1) > x(i_1, j_1)$, then $x(i_2, j_2) \geq x(i_1, j_2)$
for any $i_1 < i_2$ and $j_1 < j_2$.

We have the immediate proposition.
\begin{proposition}
\label{prop:monotone}
Fix $h$.
Then the function
	$x(s, e) = y(s, e; h) + o[s, k - 1]$
is totally monotone.
\end{proposition}

Our last step is to compute $x$ in constant time. 
This can be done by first precomputing $f[e, h]$, the log-likelihood of a segment starting
from the epoch and ending at $t(e)$ using the $h$th level. The log-likelihood of a segment is then
$y(s, e; h) = f[e, h] - f[s, h]$, which we can compute in constant time.

\begin{algorithm}[ht!]

\caption{Algorithm $\findsegments(\mathcal{P}, \Lambda)$ for finding optimal segmentation for fixed groups $\mathcal{P}$ and parameters $\Lambda$}
\label{alg:seg}

$t_{\mathit{min}} \define \min \set{t \mid (u, v, t) \in E}$\;
$f[e, h] \define $ log-likelihood of a segment $[t_{\mathit{min}}, t(e)]$ using parameters $\lambda_{\cdot\cdot h}$\;

\lForEach {$e \in E$} {
	$o[e, 1] \define \max_h f[e, h]$
}

\ForEach {$k = 2, \ldots, K$} {
	$x(s, e; h) \define o[s, k - 1] + f[e, h] - f[s, h]$\label{alg:segx}\;
	\ForEach {$h = 1, \ldots, H$} {
		$z[e, h] \define \arg \max_s x(s, e; h)$ for each $e \in E$ (use SMAWK)\;
	}
	$o[e, k] \define \max_h x(z[e, h], e; h)$ for each $e \in E$\;
	$r[e, k] \define \arg \max_h x(z[e, h], e; h)$\;
	$q[e, k] \define z[e, r[e, k]]$\;
}

Use $r$ and $q$ to
recover the optimal segmentation $(T_1, \ldots, T_K)$ and the level mapping $g$ \;

\Return $(T_1, \ldots, T_K)$, $g$\;

\end{algorithm}

The pseudo-code for finding the segmentation is given in Algorithm~\ref{alg:seg}.
A more detailed version of  the pseudo-code is given in Appendix.
Here, we first precompute $f[e, h]$.
We then solve segmentation with a dynamic program
by maintaining 3 arrays:
$o[e, k]$ is the log-likelihood
of $k$-segmentation covering the edges up to $e$, $q[e, k]$ is the starting
point of the last segment responsible for $o[e, k]$, and $r[e, k]$ is the level of the  
last segment responsible for $o[e, k]$.

In the inner loop we use SMAWK to find optimal
starting points. Note that we have to do this for each $h$, and only then select the optimal $h$ for each segment.
Note that we do define $x$ on Line~\ref{alg:segx} but we do not compute its values. Instead this function
is given to SMAWK and is evaluated in a lazy fashion.

Once we have constructed the arrays, we can recursively recover the optimal segmentation and the level mapping
from $q$ and $r$, respectively.

\findsegments runs in $\bigO{mKH + HR^2}$ time since we need to call SMAWK $\bigO{HK}$ times.

We were able to use SMAWK because
the optimization criterion turned out to be totally monotone. This was possibly only because we fixed
$\Lambda$. The notion of using SMAWK to speed up a dynamic program with totally monotone scores was proposed by~\citet{galil:90:linear}.
\citet{fleisher:06:online,hassin:91:improved} used this approach to solve dynamic program
segmenting monotonic one-dimensional sequences with $L_1$ cost. 

We fixed $\Lambda$ because Proposition~\ref{prop:nphard3} states that the optimization problem for $H < K$ cannot be solved in polynomial time if
we optimize $\mathcal{T}$, $g$, and $\Lambda$ at the same time.
Proposition~\ref{prop:nphard3} is the main reason why we cannot use directly the ideas proposed by~\citet{corneli2018multiple}
as the authors use the dynamic program to find  $\mathcal{T}$ and $\Lambda$ at the same time.

However, if $K = H$, then the
problem is solvable with a dynamic program but requires $\bigO{Km^2R^2}$ time.
However, if we consider the optimization problem as a minimization problem and shift
the cost with a constant so that it is always positive, then
using algorithms by~\citet{tatti2019strongly,guha:06:estimate}
we can obtain $(1 + \epsilon)$-approximationn with
$\bigO{K^3 \log K \log m + K^3 \epsilon^{-2} \log m}$
number of cost evaluations. Finding the optimal parameters and computing the
cost of a single segment can be done in $\bigO{R^2}$ time with $\bigO{R^2 + m}$
time for precomputing. This leads to a total time of
	$\bigO{R^2(K^3 \log K \log m + K^3 \epsilon^{-2} \log m) + m}$
for the special case of $K = H$.

\section{Related work}\label{sec:related}

The closest related work is the paper by \citet{corneli2018multiple} which
can be viewed as a special case of our approach by requiring $K = H$, in
other words, while the Poisson process may depend on time we do not take into account
any recurrent behaviour. Having $K = H$ simplifies the optimization problem
somewhat. While the general problem still remains difficult, we can now
solve the segmentation $\mathcal{T}$ \emph{and} the parameters $\Lambda$ simultaneously
using a dynamic program as was done by \citet{corneli2018multiple}.
In our problem we are forced to fix $\Lambda$ while solving the segmentation problem.
Interestingly enough, this gives us an advantage in computational time:
we only need $\bigO{KHm + HR^2}$ time to find the optimal segmentation while
the optimizing $\mathcal{T}$ and $\Lambda$ simultaneously requires $\bigO{R^2Km^2}$ time.
On the other hand, by fixing $\Lambda$ we may have a higher chance of getting stuck in a local maximum.

The other closely related work is by \citet{gionis2003finding}, where the
authors propose a segmentation with shared centroids. Here, the input is a
sequence of real valued vectors and the segmentation cost is either $L_2$ or
$L_1$ distance. Note that there is no notion of groups $\mathcal{P}$, the authors are
only interested in finding a segmentation with recurrent sources.  The authors
propose several approximation algorithms as well as an iterative method. The
approximation algorithms rely specifically on the underlying cost, in this case
$L_1$ or $L_2$ distance, and cannot be used in our case. Interestingly enough,
the proposed iterative method did not use SMAWK optimization, so it is possible
to use the optimization described in Section~\ref{sec:algorithm} to speed up
the iterative method proposed by \citet{gionis2003finding}.

In this paper, we used stochastic block model (see~\citep{Holland1983,anderson1992}, for example) as a starting point and extend it to temporal networks with recurrent sources.
Several past works have extended stochastic block models to temporal networks:
\citet{yang_2011,Matias_2015} proposed an approach where the nodes can change block memberships
over time. In a similar fashion, \citet{xu_2014} proposed a model where the adjacency matrix
snapshots are generated with a logistic function whose latent parameters evolve over time.
The main difference with our approach is that in these models the group memberships
of nodes are changing while in our case we keep the memberships constant and update
the probablties of the nodes. Moreover, these methods are based on graph snapshots
while we work with temporal edges. In another related work,
\citet{Matias_2018_estimation} modelled interactions using Poisson processes
conditioned by stochastic block model.  Their approach was to estimate the
intensities non-parametrically through histograms or kernels while we model
intensities with recurring segments. For a survey on stochastic block models,
including extensions to temporal settngs, we refer the reader to a survey
by~\citet{lee2019survey}.

Stochastic block models group similar nodes together; here similarity means
that nodes in the same group have the similar probabilities connecting to nodes
from other group.  A similar notion but a different optimization criterion was
proposed by~\citet{Arockiasamy2016combi}.  Moreover, \citet{henderson2012rolx}
proposed a method where nodes with similar neighborhoods are discovered.

In this paper we modelled the recurrency by forcing the segments to share their
parameters.  An alternative approach to discover recurrency is to look
explictly for recurrent
patterns~\citep{ozden1998,han1998MiningSP,han1999,ma_2001,yang_2003,galbrun2019}.
We should point out that these works are not design to work with graphs; instead they work
with event sequences. We leave adapting this methodology for temporal networks
as an interesting future line of work.

Using segmentation to find evolving structures in networks 
have been proposed in the past: \citet{kostakis2017} introduced a method where a temporal network
is segmented into $k$ segments with $h < k$ summaries. A summary is a graph, and the cost
of an individual segment is the difference between the summary and the snapshots in the segment. 
Moreover,~\citet{rozenshtein2020finding} proposed discovering dense subgraphs in individual segments.

\begin{table}[t!]
\setlength{\tabcolsep}{0pt}
\caption{Dataset characteristics and results from the experiments. Here, $n$ is
the number of nodes, $m$ is the number of edges, $R$ is the number of groups, $K$ is the number of segments, $H$ is the number of levels, $LL_1$ is the normalized log-likelihood for the ground truth, $G$ is the Rand index, $LL_2$ is the discovered normalized log-likelihood, $I$ is
the number of iterations, and $CT$ is the computational time in seconds. }

\label{tab:stats1}
\begin{filecontents*}{st.csv}
name,n,m,L,K,H,LL_1,R,LL_2,I,CT
\dtname{Synthetic-1},50,76332,2,2,2,0.947131044,1,0.9431811847,2,2.81s
\dtname{Synthetic-2},30,95889,3,3,3,0.9400096088,1,0.938808028,3,5.36s
\dtname{Synthetic-3},20,65056,3,3,3,0.9701469596,1,0.9693537458,3,3.91s
\dtname{Synthetic-4},60,537501,3,4,3,0.935048772,1,0.9337745296,3,23.13s
\dtname{Synthetic-5},10,33475,2,10,5,0.9070968949,1,0.9058836581,4,10.27s
\dtname{Email-Eu-1},309,61046,3,10,7,,,0.8935427203,12,188s
\dtname{Email-Eu-2},162,46772,4,8,7,,,0.8720491788,9,177s
\dtname{MathOverflow},21688,107581,2,3,2,,,0.908721286,20,263s
\dtname{CollegeMsg},1899,59835,3,8,5,,,0.8711014971,19,662s
\dtname{MOOC},7047,411749,2,3,2,,,0.8109869911,6,208s
\dtname{Bitcoin},3783,24186,3,10,10,,,0.905612076,7,115s
\dtname{Santander},735,33116,3,7,5,,,0.9404128651,20,60s
\end{filecontents*}
\pgfplotstabletypeset[
    begin table={\begin{tabular*}{\textwidth}},
    end table={\end{tabular*}},
    col sep=comma,
	columns = {name, n, m, L, K, H,LL_1, R,LL_2,I,CT},
    columns/name/.style={string type, column type={@{\extracolsep{\fill}}l}, column name=\emph{Dataset}},
    columns/n/.style={fixed, set thousands separator={\,}, column type=r, column name=$n$},
    columns/m/.style={fixed, set thousands separator={\,}, column type=r, column name=$m$},
    columns/L/.style={fixed, set thousands separator={\,}, column type=r, column name=$R$},
    columns/K/.style={fixed, set thousands separator={\,}, column type=r, column name=$K$},
    columns/H/.style={fixed, set thousands separator={\,}, column type=r, column name=$H$},
    columns/LL_1/.style={fixed, set thousands separator={\,}, column type=r, column name=$LL_1$},
    columns/R/.style={fixed, set thousands separator={\,}, column type=r, column name=$G$},
    columns/LL_2/.style={fixed, set thousands separator={\,}, column type=r, column name=$LL_2$},
    columns/I/.style={fixed, set thousands separator={\,}, column type=r, column name=$I$},
    columns/CT/.style={string type, column type=r, column name=$CT$},
    every head row/.style={
		before row={\toprule},
			after row=\midrule},
    every last row/.style={after row=\bottomrule},
]
{st.csv}
\end{table}

\section{Experimental evaluation}\label{sec:exp}

The goal in this section is to evaluate experimentally our
algorithm. Towards that end, we first test how well the algorithm discovers the ground truth using synthetic datasets.
Next we study
the performance of the algorithm 
on real-world temporal datasets in terms of running time and likelihood.
We compare our results to the following baselines: the running times are compared to a naive implementation where we do not
utilize SMAWK algorithm, and the likelihoods are compared to the likelihoods of the $(R,K)$ model.

We implemented the algorithm in Python\footnote{The source code is available at \url{https://version.helsinki.fi/chamwick/recurrent-segmentation-sbm.git}
\label{foot:code}} and performed the experiments using a 2.4GHz Intel Core i5 processor and 16GB RAM. 

\paragraph{Synthetic datasets:}
To test our algorithm, we generated  $5$ temporal networks with known groups and
known parameters $\Lambda$ which we use as a ground truth. To generate data,
we first chose a set of nodes $V$, number of groups $R$, number of segments
$K$, and  number of levels $H$. Next we assumed that each node has an equal
probability of being chosen for any group. Based on this assumption, the group
memberships were selected at random. 

We then randomly generated $\Lambda$ 
from a uniform distribution. More specifically, we generated  $H$  distinct
values for each pair of groups and map them to each segment. Note
that, we need to ensure that
each distinct level is assigned to at least one segment. To guarantee this, we
first deterministically assigned the set of $H$ levels to first $H$ segments and
the remaining ($K-H$) segments are mapped by randomly selecting ($K-H$)
elements from $H$ level set.

Given the group memberships and their related $\Lambda$, we then
generated a sequence of timestamps with a Poisson process for each pair of nodes. 
The sizes of all synthetic datasets are given in
Table~\ref{tab:stats1}.

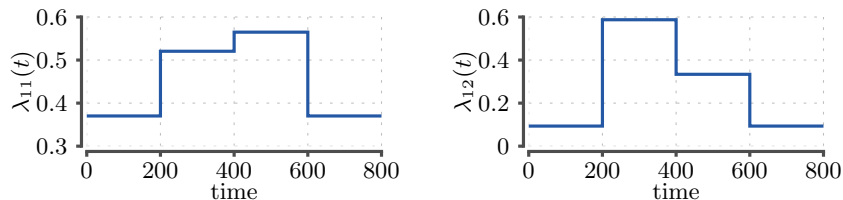
\begin{figure}[ht!]
\setlength{\tabcolsep}{10pt}
\newlength{\synfigheight}
\setlength{\synfigheight}{3.3cm}
\newlength{\synfigwidth}
\setlength{\synfigwidth}{5.5cm}
\begin{tabular}{ll}
\begin{tikzpicture}
\begin{axis}[xlabel={time}, ylabel= {$\lambda_{\mathrm{11}}(t)$},
    width = \synfigwidth,
    height = \synfigheight,
    xmin = 0,
    xmax = 800,
    ymin = 0.3,
    ymax = 0.6,
    scaled y ticks = false,
    cycle list name=yaf,
    yticklabel style={/pgf/number format/fixed},
    no markers,
]
\addplot table [x=x, y=y2, col sep=comma] {lamb-2.csv};
\pgfplotsextra{\yafdrawaxis{0}{800}{.3}{.6}}
\end{axis}
\end{tikzpicture}&

\begin{tikzpicture}
\begin{axis}[xlabel={time}, ylabel= { $\lambda_{\mathrm{12}}(t)$},
    width = \synfigwidth,
    height = \synfigheight,
    xmin = 0,
    xmax = 800,
    ymin = 0.0,
    ymax = 0.6,
    scaled y ticks = false,
    cycle list name=yaf,
    yticklabel style={/pgf/number format/fixed},
    no markers,
]
\addplot table [x=x, y=y4, col sep=comma] {lamb-2.csv};
\pgfplotsextra{\yafdrawaxis{0}{800}{.0}{.6}}
\end{axis}
\end{tikzpicture} \\

\end{tabular}

\caption{Discovered parameters $\lambda_{\mathrm{11}}(t)$,
$\lambda_{\mathrm{12}}(t)$ for the \dtname{Synthetic-4} dataset. Parameter
$\lambda_{\mathrm{12}}(t)$ implies the Poisson process parameter between group
$1$ and group $2$ as a function of time.  }

\label{fig:intensity-2}
\end{figure}

\paragraph{Real-world datasets:}
We used $7$ publicly available temporal datasets. \dtname{Email-Eu-1} and \dtname{Email-Eu-2} are collaboration networks between researchers in a European research institution.\!\footnote{\url{http://snap.stanford.edu}\label{foot:snap}}
\dtname{Math Overflow} contains user interactions in  Math Overflow web site while answering to the questions.\!\footref{foot:snap}
\dtname{CollegeMsg} is an online message network at the University of California, Irvine.\!\footref{foot:snap}
\dtname{MOOC} contains actions by users  of a popular MOOC platform.\!\footref{foot:snap}
\dtname{Bitcoin} contains member rating interactions  in a bitcoin trading platform.\!\footref{foot:snap}
\dtname{Santander} contains station-to-station links that occurred on Sep $9$, $2015$  from the Santander bikes hires in  London.\!\footnote{\url{https://cycling.data.tfl.gov.uk}} The sizes of these networks are given in Table~\ref{tab:stats1}.

\paragraph{Results for synthetic datasets:}

To evaluate the accuracy of our algorithm, we compare the set of
discovered groups and their intensity functions with the ground truth groups
and intensity functions. Our algorithm found exact 
groups of nodes: in
Table~\ref{tab:stats1} we can see that Rand index
\!\footnote{The Rand index is used to measure the similarity between two
groups~\citet{rand1971objective}.
The Rand index ranges between $0$ and $1$, and when the group partitions agree perfectly, the Rand index  is $1$. For a complete disagreement, the Rand index is $0$.\label{foot:rand}}
(column $G$) is equal to $1$.
 
Next we compare the log-likelihood values from true models against the log-likelihoods of discovered models. To
evaluate the log-likelihoods, we normalize the log-likelihood, that is we computed
	$\lhood{\mathcal{P}, \mathcal{T}, g, \Lambda}/\lhood{\mathcal{P}', \mathcal{T}', g', \Lambda'}$,
where $\mathcal{P}', \mathcal{T}', g', \Lambda'$ is a model with a single group and a single segment.
Since all our log-likelihood values were negative,
the \emph{normalized log-likelihood} values were between $0$ and $1$, and \emph{smaller} values are better.

As demonstrated in column $LL_1$ and column $LL_2$ of
Table~\ref{tab:stats1}, we obtained similar normalized log-likelihood
values when compared to the normalized log-likelihood of the ground truth.
The obtained normalized log-likelihood values were all slightly better
than the log-likelihoods of the generated models, that is, our solution is as good as the ground truth.

An example of
the
discovered parameters, $\lambda_{11}$ and $\lambda_{12}$, for \dtname{Synthetic-4} dataset are shown in
Figure~\ref{fig:intensity-2}.
The discovered parameters matched closely to the generated parameters with the biggest absolute difference being $0.002$ for \dtname{Synthetic-4}.
The figures for other values and other synthetic datasets are similar.

\paragraph{Computational time:}

Next we consider the  computational  time  of  our  algorithm. 
We varied the parameters $R$,
$K$, and $H$ for each dataset. The model parameters and computational times are given in
Table~\ref{tab:stats1}. From the last column
$CT$, we see that the running times are reasonable despite using inefficient
Python libraries: for example we were able to compute the model for \dtname{MOOC}
dataset, with over $400\,000$ edges, under four minutes.
This implies that the algorithm scales well for large networks.
This is further supported by a low number of iterations, column $I$ in Table~\ref{tab:stats1}.

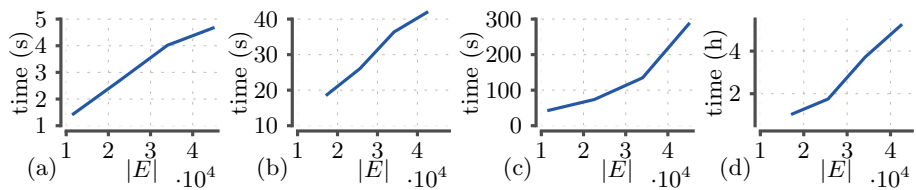
\begin{figure}[t!]
\begin{center}
\begin{tabular}{llll}
\begin{tikzpicture}
\begin{axis}[xlabel={$\abs{E}$}, ylabel= {time (s)},
    width = 3.6cm,
	height = 3cm,
    xmin = 10000,
    xmax = 46000,
    ymin = 1,
    ymax = 5,
    scaled y ticks = false,
    cycle list name=yaf,
    yticklabel style={/pgf/number format/fixed},
    no markers,
]
\addplot table [x=x, y=y, col sep=comma] {edges-synthetic.csv};
\pgfplotsextra{\yafdrawaxis{10000}{46000}{1}{5}}
\end{axis}
\node[anchor=north east] at (0, -0.3) {(a)};
\end{tikzpicture}&
\begin{tikzpicture}
\begin{axis}[xlabel={$\abs{E}$}, ylabel= {time (s)},
    width = 3.6cm,
	height = 3cm,
    xmin = 10000,
    xmax = 48000,
    ymin = 10,
    ymax = 40,
    scaled y ticks = false,
    cycle list name=yaf,
    yticklabel style={/pgf/number format/fixed},
    no markers,
]
\addplot table [x=x, y=y, col sep=comma] {san-smawk.csv};

\pgfplotsextra{\yafdrawaxis{10000}{48000}{10}{40}}
\end{axis}
\node[anchor=north east] at (0, -0.3) {(b)};
\end{tikzpicture} &
\begin{tikzpicture}
\begin{axis}[xlabel={$\abs{E}$}, ylabel= {time (s)},
    width = 3.6cm,
	height = 3cm,
    xmin = 10000,
    xmax = 46000,
    ymin = 0,
    ymax = 300,
    scaled y ticks = false,
    cycle list name=yaf,
    yticklabel style={/pgf/number format/fixed},
    no markers,
]
\addplot table [x=x, y=y, col sep=comma] {syn-algo-3.csv};
\pgfplotsextra{\yafdrawaxis{10000}{46000}{0}{300}}
\end{axis}
\node[anchor=north east] at (0, -0.3) {(c)};
\end{tikzpicture}&
\begin{tikzpicture}
\begin{axis}[xlabel={$\abs{E}$}, ylabel= {time (h)},
    width = 3.6cm,
	height = 3cm,
    xmin = 10000,
    xmax = 45000,
    ymin = .5,
    ymax = 5.5,
    scaled y ticks = false,
    cycle list name=yaf,
    yticklabel style={/pgf/number format/fixed},
    no markers,
]
\addplot table [x=x, y=y, col sep=comma] {san-naive.csv};

\pgfplotsextra{\yafdrawaxis{10000}{45000}{.5}{5.5}}
\end{axis}
\node[anchor=north east] at (0, -0.3) {(d)};
\end{tikzpicture}
\end{tabular}

\end{center}
\caption{Computational time as a function of number of temporal edges ($\abs{E}$) for
\dtname{Synthetic-large} (a,c) and \dtname{Santander-large} (b,d). This
experiment was done with  $R=3$, $K=5$, and $H=3$ using SMAWK algorithm (a--b) and naive dynamic programming (c--d).
The times are in seconds in (a--c) and in hours in (d).
}
\label{fig:edgestime}
\end{figure}

Next we study the computational time as a function of $m$,
number of edges.

We first prepared $4$ datasets with different number of edges from a real-world
dataset; \dtname{Santander-large}. To vary the number of edges, we uniformly sampled
edges without replacement. We sampled like a $.4$, $.6$,
$.8$, and $1$ fraction of edges.

Next we created $4$ different  \dtname{Synthetic-large} dataset with 30
nodes, 3 segments with unique $\lambda$ values but with different
number of edges. To do that, we  gradually increase the number of Poisson
samples we generated for each segment.

From the results in Figure~\ref{fig:edgestime} we see that generally
computational time increases as $\abs{E}$ increases. For instance, a set of
$17\,072$  edges accounts for $18.46$s whereas  a set of $34\,143$ edges accounts
for $36.36$s w.r.t \dtname{Santander-large}. Thus a linear trend w.r.t $\abs{E}$ is
evident via this experiment.

To emphasize the importance of SMAWK, we replaced it with a stock
solver of the dynamic program, and repeat the experiment. We observe
in Figure~\ref{fig:edgestime} that computational time has increased
drastically when stock dynamic program algorithm is used. For example, a set of
$34\,143$ edges required $3.7$h for \dtname{Santander-large} dataset but only
$36.36$s when SMAWK is used.

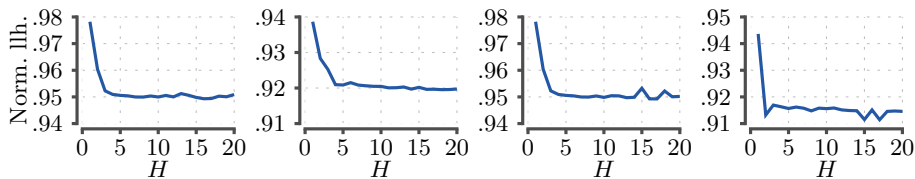
\begin{figure}[t!]
\setlength{\tabcolsep}{0pt}
\begin{center}
\begin{tabular}{llll}

\begin{tikzpicture}
\begin{axis}[xlabel={$H$}, ylabel= {Norm. llh.},
    width = 3.6cm,
	height = 3cm,
    xmin = 0,
    xmax = 20,
    ymin = .94,
    ymax = .98,
    scaled y ticks = false,
    cycle list name=yaf,
    yticklabel style={/pgf/number format/fixed, /pgf/number format/skip 0.},
    no markers,
]
\addplot table [x=x, y=y, col sep=comma] {h-like-20.csv};
\pgfplotsextra{\yafdrawaxis{0}{20}{0.94}{0.98}}
\end{axis}
\end{tikzpicture}&
\begin{tikzpicture}
\begin{axis}[xlabel={$H$}, 
    width = 3.6cm,
	height = 3cm,
    xmin = 0,
    xmax = 20,
    ymin = .91,
    ymax = .94,
    scaled y ticks = false,
    cycle list name=yaf,
    yticklabel style={/pgf/number format/fixed, /pgf/number format/skip 0.},
    no markers,
]
\addplot table [x=x, y=y1, col sep=comma] {h-like-20.csv};

\pgfplotsextra{\yafdrawaxis{0}{20}{0.91}{.94}}
\end{axis}
\end{tikzpicture} & 
\begin{tikzpicture}
\begin{axis}[xlabel={$H$}, 
    width = 3.6cm,
	height = 3cm,
    xmin = 0,
    xmax = 20,
    ymin = .94,
    ymax = .98,
    scaled y ticks = false,
    cycle list name=yaf,
    yticklabel style={/pgf/number format/fixed, /pgf/number format/skip 0.},
    no markers,
]
\addplot table [x=x, y=y2, col sep=comma] {h-like-20.csv};
\pgfplotsextra{\yafdrawaxis{0}{20}{0.94}{0.98}}
\end{axis}
\end{tikzpicture}&
\begin{tikzpicture}
\begin{axis}[xlabel={$H$}, 
    width = 3.6cm,
	height = 3cm,
    xmin = 0,
    xmax = 20,
    ymin = .91,
    ymax = .95,
    scaled y ticks = false,
    cycle list name=yaf,
    yticklabel style={/pgf/number format/fixed, /pgf/number format/skip 0.},
    no markers,
]
\addplot table [x=x, y=y3, col sep=comma] {h-like-20.csv};
\pgfplotsextra{\yafdrawaxis{0}{20}{0.91}{.95}}
\end{axis}
\end{tikzpicture}

\end{tabular}
\end{center}
\caption{Normalized log-likelihood as a function of number of levels ($H$) for the \dtname{Santander} dataset (top-left), \dtname{bitcoin} dataset (top-right), \dtname{Synthetic-5} dataset (bottom-left), and \dtname{Email-Eu-1} dataset (bottom-right). This experiment is done for  $R=2$, $K=20$, and $H=1, \ldots, 20$.}
\label{fig:h-like}
\end{figure}

\paragraph{Likelihood vs number of levels:}
 
Our next experiment is to study how normalized log-likelihood behaves upon the
choices of $H$. We conducted this experiment for  $K=20$ and vary the number of
levels~($H$) from $H=1$ to $H=20$.
The results for the \dtname{Santander}, \dtname{Bitcoin}, \dtname{Synthetic-5},
and \dtname{Email-Eu-1} dataset are shown in Figure~\ref{fig:h-like}. From the
results we see that  generally normalized log-likelihood decreases as $H$
increases. That is due to the fact that higher the $H$ levels, there exists a
higher degree of freedom in terms of optimizing the likelihood.
Note that if $H = K$, then our model corresponds to the model studied by~\citet{corneli2018multiple}.
Interestingly enough,
the log-likelihood values plateau for values of $H \ll K$ suggesting that existence
of recurring segments in the displayed datasets.

\section{Concluding remarks}\label{sec:conclusions}

In this paper we introduced a problem of finding recurrent sources in temporal
network: we introduced stochastic block model with recurrent
segments.

To find good solutions we introduced an iterative algorithm by considering 3
subproblems, where we optimize blocks, model parameters, and segmentation in
turn while keeping the remaining structures fixed.  We demonstrate how each
subproblem can be optimized in $\bigO{m}$ time. Here, the key step is to
use SMAWK algorithm for solving the segmentation. This leads to a computational
complexity of $\bigO{KHm + Rn + R^2H}$ for a single iteration.
We show experimentally that the number of iterations is low, and that the
algorithm can find the ground truth using synthetic datasets.

The paper introduces several interesting directions: \citet{gionis2003finding}
considered several approximation algorithms but they cannot be applied directly
for our problem because our optimization function is different. Adopting these
algorithms in order to obtain an approximation guarantee is an interesting
challenge.  We used a simple heuristic to optimize the groups. We chose this
approach due to its computational complexity. Experimenting with more
sophisticated but slower methods for discovering block models, such as methods
discussed in~\citep{abbe2017community}, provides a fruitful line of future
work.

\bibliographystyle{splncs04nat}
\bibliography{bibliography}

\appendix
\section{Proofs}
\begin{proof}[Proof of Proposition~\ref{prop:nphard1}]
Assume that we given an instance of \textsc{MaxCut}, that is, a static graph $H$ with $n$ nodes and $m \geq n$ edges.
Define
\[
r = 2^6(n + 1)^2\quad\text{and}\quad
\alpha = \max(2^{20} m^2, r^2)\quad.
\]
The temporal graph $G$ consists of two copies of $H$: we will denote the nodes of the
copies with $U = u_1, \ldots, u_n$ and $V = v_1, \ldots, v_n$. We connect the nodes in $U$ (and $V$) to match the edges in $H$ at timestamp $1$.
We connect the corresponding nodes $u_i$ and $v_i$ with $n\alpha$ cross edges at timestamp $0$.
We also add two sets of $r$ nodes, which we will denote by $X$ and $Y$, and connect
each node pair $(x, y)$, where $x \in X$ and $y \in Y$, with $\alpha$ edges at timestamp $0$.

We set $H = K = 1$, which forces the segmentation to be a single segment $[0, 1]$.

Set $R = 2$. Let $\mathcal{P} = \set{P_1, P_2}$ be the optimal solution, and let $\Lambda$ be its parameters.

We will prove in Lemma~\ref{lem:sym} that $X \subseteq P_1$ and $Y \subseteq P_2$ or $X \subseteq P_2$ and $Y \subseteq P_1$.
Moreover, $u_i \in P_1$ implies that $v_i \in P_2$ and $v_i \in P_1$ implies $u_i \in P_2$.

This immediately implies that $\lambda_{11} = \lambda_{22}$. Moreover, since $r > n$, we have
$\lambda_{12} > \alpha /4 > 1 \geq \lambda_{11}$. 

Let us define $C_i = U \cap P_i$ and $D_i = V \cap P_i$. Write $x$ to be the number of cross edges
between $C_1$ and $C_2$.

Let $C_1' \cup C_2' = U$ be the maximum cut, and let $D_1' \cup D_2' = V$ be the corresponding cut in $V$.
Define $\mathcal{P'} = \set{X \cup C'_1  \cup D'_2,  Y \cup C'_2  \cup D'_1}$.
Write $x'$ to be the number of cross edges between $C_1'$ and $C_2'$. By optimality $x' \geq x$.

The log-likelihood of $\mathcal{P}'$ is
\[
\begin{split}
	\lhood{\mathcal{P}'} & \geq \lhood{\mathcal{P}', \Lambda} \\
	& = (2m - 2x') \log \lambda_{11} + 2x' \log \lambda_{12} + Z \\
	& \geq (2m - 2x) \log \lambda_{11} + 2x \log \lambda_{12} + Z \\
	& = \lhood{\mathcal{P}} \geq  \lhood{\mathcal{P}'}, \\
\end{split}
\]
where
\[
	Z = \alpha(n^2 + r^2)\log \lambda_{12} - 2m - \alpha(n^2 + r^2)\quad.
\]
We have shown that $x' = x$, proving the NP-hardness of finding $\mathcal{P}$ with the optimal likelihood.
\end{proof}

\begin{lemma}
\label{lem:sym}
Let $\mathcal{P}$ be the partition as defined in the proof of Proposition~\ref{prop:nphard1}.
Then $X \subseteq P_1$ and $Y \subseteq P_2$ or $X \subseteq P_2$ and $Y \subseteq P_1$.
Moreover, $u_i \in P_1$ implies that $v_i \in P_2$ and $v_i \in P_1$ implies $u_i \in P_2$.
\end{lemma}
\begin{proof}
To prove the lemma we will need several counters: let us define
$a_i = \abs{P_i \cap X}$,
$b_i = \abs{P_i \cap Y}$,
$c_i = \abs{P_i \cap U}$, and
$d_i = \abs{P_i \cap V}$.
We also write $x_i = a_i + b_i$, $y_i = c_i + d_i$ and $z_i = x_i + y_i$.

Define $k_{ij}$ to be the number of cross edges between $U$ and $V$ in $(i, j)$th block of $\mathcal{P}$.
Similarly, let $m_{ij}$ to be the number of edges in $U$ and $V$ (that is, the cross edges are excluded) in $(i, j)$th block of $\mathcal{P}$.
Note that $k_{12} + k_{11} + k_{22} = n$ and $m_{12} + m_{11} + m_{22} = 2m$.

The parameters for $\mathcal{P}$ are
\[
	\lambda_{12} = \frac{\alpha a_1 b_2 + \alpha a_2 b_1 + \alpha n k_{12} + m_{12}}{z_1z_2},
\]
\[
	\lambda_{11} = \frac{\alpha a_1 b_1 + \alpha n k_{11} + m_{11}}{{z_1 \choose 2}},
\]
and
\[
	\lambda_{22} = \frac{\alpha a_2 b_2 + \alpha n k_{22} + m_{22}}{{z_2 \choose 2}}\quad.
\]

Let us define $\mathcal{P}' = \set{X \cup U, Y \cup V}$.
Note that the parameters for $\mathcal{P}'$ are equal to
\[
    \lambda_{12}' = \frac{\alpha n^2 + \alpha r^2}{(r + n)^2} 
\quad\text{and}\quad
    \lambda_{11}' = \lambda_{22}' = \frac{m}{{n + r \choose 2}}\quad. 
\]

To prove the claim we will assume that $a_1 b_1 + a_2 b_2  > 0$ or $k_{12} < n$, and
show that $\lhood{\mathcal{P}'} > \lhood{\mathcal{P}}$ which is a contradiction.

First, note that we can write the score difference as
\[
	\lhood{\mathcal{P}'} - \lhood{\mathcal{P}} = A + B + C,
\]
where
\[
	A = \alpha(a_1 b_2 + a_2 b_1 + n k_{12}) \log \frac{\lambda_{12}'}{\lambda_{12}},
\]

\[
	B = \alpha (a_1b_1 + n k_{11}) \log \frac{\lambda_{12}'}{\lambda_{11}} + \alpha (a_2b_2 + n k_{22}) \log \frac{\lambda_{12}'}{\lambda_{22}},
\]
and
\begin{equation}
\label{eq:partc}
	C = m_{11} \log \frac{\lambda_{11}'}{\lambda_{11}} + m_{22} \log \frac{\lambda_{11}'}{\lambda_{22}} + m_{12} \log \frac{\lambda_{11}'}{\lambda_{12}}\quad.
\end{equation}

We claim that 
\[
	A \geq 0,\quad 
	B > \alpha 2^{-7},
	\quad\text{and}\quad
	C \geq 2m \log \frac{1}{4\alpha n r^2}\quad.
\]

This proves the lemma since
\[
\begin{split}
	B + C & > \alpha 2^{-7} - 2m \log 4 \alpha n r^2 \\
	& \geq \alpha 2^{-7} - 8m \log \alpha \\
	& \geq \alpha 2^{-7} - 8m \sqrt{\alpha} \\
	& = \sqrt{\alpha} (\sqrt{\alpha}2^{-7} - 8m ) \geq 0\quad.
\end{split}
\]

We will first bound $C$.
Since we may have at most $\alpha n$ edges per node pair, we have $\lambda_{11}, \lambda_{12}, \lambda_{22} \leq \alpha n$.
Moreover, since $r > n$, we have $\lambda_{11}' \geq (r + n)^{-2} \geq r^{-2}/4$. The bound follows from Eq.~\ref{eq:partc}.

Next we will bound $B$.
Assume that $x_1, x_2 \geq 2n$.
Our next step is to upper bound $\lambda_{11}$ and $\lambda_{22}$.
In order to do this, first note that since $m_{11} \leq \alpha / 2$ and $k_{11} \leq y_1$, we have
\[
	\frac{m_{11} + \alpha n k_{11}}{{y_1 \choose 2} + y_1 x_1} \leq \alpha \frac{1/2 + n y_1 }{{y_1 \choose 2} + 2n y_1 } \leq \alpha \frac{1/2 + n y_1 }{1 + 2n y_1 } = \alpha / 2\quad.
\]
In addition, since $x_1 \geq 3$, we have
\[
	\alpha \frac{a_1b_1}{{x_1 \choose 2}} \leq 2\alpha / 3\quad.
\]
We can combine the two bounds, leading to
\[
\begin{split}
	\lambda_{11} & = \frac{\alpha a_1b_1 + m_{11} + \alpha n k_{11}}{{z_1 \choose 2}} 
	 = \frac{a_1b_1 + m_{11} + \alpha n k_{11}}{{x_1 \choose 2} + {y_1 \choose 2} + y_1 x_1} \leq 2\alpha/3\quad.
\end{split}
\]
The same bound holds for $\lambda_{22}$.

Since $r \geq 4n$, we have $\lambda_{12}' \geq 17\alpha/25$. Thus,
\[
	B \geq \alpha (a_1b_1 + a_2b_2 + nk_{11} + nk_{22}) \log \frac{17 \times 3}{25 \times 2} > \alpha 2^{-7}\quad.
\]

Assume now that $x_1 < 2n$. Then $a_2 + b_2 = x_2 > 2r - 2n > 1.5 r$. Since $a_2, b_2 \leq r$, we must have $a_2, b_2 \geq r/2$.
Moreover, using the previous arguments, we have
\[
\log \frac{\lambda_{12}'}{\lambda_{22}} \geq 2^{-7}
\quad\text{and}\quad
\log \frac{\lambda_{12}'}{\lambda_{11}} \geq \log \frac{17}{25n}
\quad.
\]
Consequently,
\[
\begin{split}
	B & =
	\alpha (a_2b_2 + nk_{22}) \log \frac{\lambda_{12}'}{\lambda_{22}} 
	+ \alpha (a_1b_1 + nk_{11}) \log \frac{\lambda_{12}'}{\lambda_{11}} \\
	& \geq \alpha \pr{\frac{r^2}{2^9} - 4 n^2 \log \frac{25 n}{17}} \\
	& \geq \alpha (8(n + 1)^4 - 4 n^2 \log 2 n) \\
	& \geq \alpha (8(n + 1)^4 - 8 n^{4}) \geq \alpha > \alpha 2^{-7}\quad.
\end{split}
\]

Finally we will bound $A$ by showing that $\lambda_{12} \leq \lambda'_{12}$.
Assume for simplicity that $x_1 \leq x_2$.
Let us define
\[
	N = a_1 b_2 + a_2 b_1 + 1 + nk_{12}\quad.
\]
Assume that $z_1, z_2 \geq r$.
We claim that $N \leq r^2 + (z_1 - r)(z_2 - r)$, which leads to
\[
	\lambda_{12} \leq \alpha\frac{N}{z_1z_2} \leq \alpha\frac{r^2 + (z_1 - r)(z_2 - r)}{z_1z_2} = \alpha - \alpha \frac{2nr}{z_1z_2}\quad.
\]
Here the first inequality holds since $m_{12} \leq 2m \leq \alpha$.
The right hand side achieves its maximum when $z_1z_2$ is maximized, that is, $z_1
= z_2 = n + r$. In such case, the upper bound is equal to $\lambda_{12}'$.

To prove the claim, first note that
\begin{equation}
\label{eq:sep}
	nk_{12} \leq \max(y_1, y_2) \min(y_1, y_2) = y_1y_2
\end{equation}
with the equality holding if and only if $k_{12} = y_1 = y_2 = n$.

Assume that $x_1 = x_2 = r$.
If $a_1b_1 + a_2b_2 > 0$, then
\[
\begin{split}
	N & = r^2 - a_1b_1 + a_2b_2 + 1  + nk_{12} 
	 \leq r^2  + y_1y_2 = r^2 + (z_2 - r)(z_1 - r)\quad.
\end{split}
\]
Assume $a_1b_1 + a_2b_2 = 0$.
Then $k_{12} < n$, and the inequality is strict in Eq.~\ref{eq:sep}. Consequently, 
\[
\begin{split}
	N & = r^2 + n k_{12} + 1 
	\leq r^2 + y_1y_2 = r^2 + (z_2 - r)(z_1 - r) \quad .\\
\end{split}
\]

Assume now that $x_1 \leq x_2 - 1$. Since a node in $(X \cup Y) \cap P_1$ is connected to $r$ nodes, we must have
$a_1 b_2 + a_2 b_1 \leq x_1r$. Due to the assumption,
$x_1 \leq r - 1$ and $x_2 \geq r + 1$, which leads to 
\[
\begin{split}
	k_{12} & \leq \min(y_1, y_2) \\
	& = \min(z_1 - x_1, z_2 - x_2) \\
	& \leq \min(z_1 - x_1, z_2 - r)
\end{split}
\]
and
\[
\begin{split}
	N & \leq x_1 r + 1 + n \min(z_1 - x_1, z_2 - r) \\
	& \leq (r - 1)r + 1 + n \min(z_1 - r + 1, z_2 - r) \\
	& \leq r^2 + (1 + n - r) + n \min(z_1 - r, z_2 - r) \\
	& \leq r^2 + \max(z_1 - r, z_2 - r)\min(z_1 - r, z_2 - r) \\
	& = r^2 + (z_2 - r)(z_1 - r)\quad.
\end{split}
\]

As a final case assume that $z_1 < r$.
If $x_1 = z_1$, then
\[
\begin{split}
	\lambda_{12} & \leq \alpha \frac{rz_1}{z_1 z_2} \leq \alpha \frac{r}{r + 2n} 
	 = \alpha \frac{r^2}{r^2 + 2nr} \leq  \alpha \frac{r^2 + n^2}{r^2 + 2nr + n^2} = \lambda_{12}'\quad.
\end{split}
\]
If $x_1 < z_1$, then
\[
	N \leq x_1r + 1 + n (z_1 -  x_1) \leq z_1r
\]
and again
\[
	\lambda_{12} \leq \alpha \frac{N}{z_1z_2} \leq \frac{zr_1}{z_1z_2} \leq \lambda_{12}'\quad.
\]
The case for $z_2 < r$ is symmetrical.

We have now proven our claim, and thus proved that $\lambda_{12} \leq \lambda_{12}'$. Consequently, $A \geq 0$.
\end{proof}

\begin{proof}[Proof of Proposition~\ref{prop:nphard2}]
To prove \np-hardness we will reduce the \textsc{MaxCut} problem, where
we are asked to partition graph into $2$ subgraphs and maximize cross-edges.

Assume that we given a static graph $H$. We will use $H$ as our temporal graph
$G$ by setting the edges to the same timestamp, say $t$. We also set $H = K = 1$, and
use $R = 2$ groups. We also set the segmentation $\mathcal{T} = {[t, t]}$.

Select two values $\alpha < \beta$ and set the parameters
$\lambda_{11} = \lambda_{22} = \alpha$ and $\lambda_{12} = \beta$.

Let $P_1, P_2$ be a partition of the nodes and let $x$ be the number
of the inner edges, that is, edges $(u, v, t)$ with $u, v \in P_1$ or $u, v \in P_2$.
Note that $m - x$ is the number of cross edges.

The log-likelihood is then equal to
\[
    \lhood{\mathcal{P}, \mathcal{T}, g, \Lambda} =
    x \log \alpha + (m - x) \log \beta,
\]
which is maximized when $m - x$ is maximized since $\beta > \alpha$.
Since $m - x$ is the number of cross-edges, this completes the proof.
\end{proof}

\begin{proof}[Proof of Proposition~\ref{prop:gain}]
Let us write $\mathcal{Q}$ to be the partition obtained from $\mathcal{P}$ by
deleting $u$, that is $Q_b = P_b \setminus \set{u}$ and $Q_j = P_j$ for $j \neq b$.
Note that
\[
    \lhood{P_a', P_j', T, \lambda} - \lhood{Q_a, Q_j, T, \lambda} = \lhood{u, P_j', T, \lambda}
\]
for any $T$ and $\lambda$. Moreover, $\lhood{P_i', P_j', T, \lambda} = \lhood{Q_i, Q_j, T, \lambda}$
if $i, j \neq a$.
The score difference is equal to
\[
\begin{split}
 & \lhood{\mathcal{P}', \mathcal{T}, g, \Lambda} - \lhood{\mathcal{Q}, \mathcal{T}, g, \Lambda} \\
 & \  = \sum_{i \leq j} \sum_{k = 1}^K \lhood{P_i', P_j', T_k, \lambda_{ijg(k)}} - \lhood{Q_i, Q_j, T_k, \lambda_{ijg(k)}} \\
 & \  = \sum_{j = 1}^R \sum_{k = 1}^K \lhood{P_a', P_j', T_k, \lambda_{ajg(k)}} - \lhood{Q_a, Q_j, T_k, \lambda_{ajg(k)}} \\
 & \  = \sum_{j = 1}^R \sum_{k = 1}^K \lhood{u, P_j', T_k, \lambda_{ajg(k)}}  \\
 & \  = \sum_{j = 1}^R \sum_{h = 1}^H \lhood{u, P_j', \mathcal{S}_h, \lambda_{ajh}} \\
 & \  = \sum_{j = 1}^R \sum_{h = 1}^H \cnt{u, P_j', \mathcal{S}_h} \log  \lambda_{ajh} - t_h \abs{Q_j} \lambda_{ajh} \\
 & \  = \sum_{j = 1}^R \sum_{h = 1}^H c_{jh} \log  \lambda_{ajh} - t_h \abs{Q_j} \lambda_{ajh} \\
 & \  = \sum_{h = 1}^H \lambda_{bah} t_h + \sum_{j = 1}^r c_{jh} \log  \lambda_{ajh} - t_h \abs{P_j} \lambda_{ajh}\quad. \\
\end{split}
\]
Here we used the fact that $\cnt{u, P_j, \mathcal{S}_h} = \cnt{u, P_j', \mathcal{S}_h}$. The claim follows by
setting $Z = \lhood{\mathcal{Q}, \mathcal{T}, g, \Lambda} - \lhood{\mathcal{P}, \mathcal{T}, g, \Lambda}$.
\end{proof}

\begin{proof}[Proof of Proposition~\ref{prop:nphard3}]
Assume that we given an instance of \textsc{3-Matching}, that is, a domain $X$
of size $n$, where $n$ is divisible by 3, and a collection $\fm{S}$ of $m$ sets such that $S \subseteq X$ and $\abs{S} = 3$ for each $S \in \fm{S}$.
The problem whether there is a disjoint subcollection in $\fm{S}$ covering $X$
is known to be \np-complete.

Let $\fm{T} = \set{S \subseteq X \mid \abs{S} = 3, S \notin \fm{S}}$ be the complement collection of $\fm{S}$.
For each $i \leq j \leq m$, define $c_{ij}$
to be the number of sets in $S$ containing $i$ and $j$,
\[
	c_{ij} = \abs{\set{S \in \fm{S} \mid \set{i, j} \subset S}}\quad.
\]

To construct the dynamic graph $G$ we will use 5 sets of nodes, namely $\set{u}$, $A$, $B$, $C$, $D$. 
The first set consists only of one node $u$. Every edge will be adjacent to $u$.
The second set $A$ contains as many nodes as there are sets in $\fm{T}$.
For each $i \in T_j \in \fm{T}$, we add an edge $(u, a_j)$ at timestamp $i$.
The third set $B$ contains $\sum_{i < j} c_{ij}$ nodes
which we divide further into $n(n - 1)/2$ sets $B_{ij}$ with $\abs{B_{ij}} = c_{ij}$.
For each $i < j$ we connect
nodes in $B_{ij}$ with $u$ at timestamp $i$ and at timestamp $j$.
The fourth set $C$ contains $n(m - c_{ii})$ nodes
which we divide further into $n$ sets $C_{i}$ with $\abs{C_{i}} = m - c_{ii}$.
For each $i \leq n$ we connect
nodes in $C_i$ with $u$ at timestamp $i$.
The fifth set $D$ contains $nw$ nodes, where $w = 24n^6$,
which we divide further into $n$ sets $D_{i}$ with $\abs{D_{i}} = w$.
For each $i \leq n$ we connect
$w$ nodes with $u$ at timestamp $i$.

We will set $K = n$ and $H = n / 3$. We set $R$ to be the number of nodes and set the partition $\fm{P}$
to be the partition where each node is contained in its own block. We require for a segmentation to start from $0$.
Since there are only $n$ timestamps, the segmentation consists of $n$ segments of form $(i - 1, i]$ or $[0, 1]$.

Let $g$ be the optimal grouping of segments and let $\Lambda$ be its
parameters.  We first claim that $g$ groups timestamps into groups of 3. To
prove this assume that there is a group of size $y = 1,2$. Then there is
another group with a size of $x \geq 6 - y$.  Let $g'$ be a mapping where we
move $3 - y$ timestamps from the larger group to the smaller group and let $\Lambda'$ be the new optimal parameters.

The number of edges adjacent to $A$, $B$, and $C$ can be bound by 
$3n^3 +2n^2m + nm \leq 6n^5$. Moreover, the non-zero parameters can be bound by 
$\lambda' \geq 1/n$ and $\lambda \leq 1$. Consequently, the score difference 
can be bound by
\[
	\lhood{g'} - \lhood{g} \geq 6n^5(\log 1/n - \log 1) + Z(x) \geq 6n^6 + Z(x),
\]
where $Z(x)$ is equal to
\[
	w ((x - 3 + y)\log \frac{1}{x - 3 + y} + 3 \log \frac{1}{3} - x\log \frac{1}{x} - y \log \frac{1}{y})\ .
\]

The derivative of $Z(x)$ with respect to $x$ is equal to $\log(x / (x - 3 + y)) > 0$, that is, $Z(x)$ is at smallest when $x = 6 - y$.
A direct calculation shows that $Z(x)$ is the smallest when $y = 2$, leading to
\[
\begin{split}
	\lhood{g'} - \lhood{g} & \geq -6n^6 + w (6 \log \frac{1}{3} - 4\log \frac{1}{4} - 2 \log \frac{1}{2}) \\
	& > -6n^6 + w/4  = 0\quad.\\
\end{split}
\]
In summary, $\lhood{g'} > \lhood{g}$ which is a contradiction. Thus, $g$ groups
of segments to size of at least $3$.  Since $K = 3H$, the groups are exactly of
size 3.

Our next step is to calculate the impact of a single group to the score. In order to do that
first note that for any $i < j$ there are
\[
	{n \choose 3} - c_{ij} + c_{ij} = {n \choose 3}
\]
edges joining the same nodes at timestamp $i$ and timestamp $j$.
Similarly, there are
\[
	{n \choose 3} - c_{ii} + (\sum_{i < j} c_{ij}) + m - c_{ii} + w = {n \choose 3} + m + w
\]
edges adjacent to node $i$.

Consider a set $S$ of size $3$ induced by $g$. Assume that $S \notin \fm{T}$.
Then there are $3{n \choose 3}$ parameters in $\Lambda$ assosiated with the group with value $2/3$,
and $3(m + w)$ parameters with value $1/3$. The remaining parameters are 0. Consequently, the impact to the score is equivalent to
\[
	\alpha = 6 {n \choose 3} \log 2/3 + 3(m + w) \log 1/3\quad.
\]
Assume that $S \in \fm{T}$. Then using exclusion-inclusion principle,
there are $3{n \choose 3} - 3$ parameters with value $2/3$,
$3(m + w) + 3$ parameters with value $1/3$, one parameter with value $3/3$
and the remaining parameters are 0. Consequently, the impact to the score is equivalent to
\[
	\beta = 6{n \choose 3} \log 2/3 - 6\log 2/3 + 3(1 + m + w) \log 1/3\quad.
\]
We immediately see that
\[
	\alpha - \beta = 6 \log 2/3 - 3 \log 1/3 = 3 \log 4/3 > 0\quad.
\]
Let $k$ be the number of groups induced by $g$ that are in $\fm{S}$. Then
the score is equal to 
\[
	\lhood{g} = k \alpha + (H - k) \beta - \abs{E(G)}.
\]
Since $\alpha > \beta$, there is a disjoint subcollection in $\fm{S}$ covering $X$
if and only if $\lhood{g} = H \alpha - \abs{E(G)}$.
\end{proof}

\begin{proof}[of Proposition~\ref{prop:monotone}]
Assume four
edges $s_1$, $s_2$, $e_1$ and $e_2$ with $t(s_1) \leq t(s_2)$
and $t(e_1) \leq t(e_2)$. We can safely assume that $t(s_2) \leq t(e_1)$.
We can write the difference as
\[
    x(s_1, e_2) - x(s_1, e_1) = y(e_1, e_2; h) =
    x(s_2, e_2) - x(s_2, e_1)
    \ .
\]
Thus, if $x(s_2, e_1) > x(s_1, e_1)$, then $x(s_2, e_2) > x(s_1, e_2)$,
completing the proof.
\end{proof}

\newpage

\section{Detailed pseudo-code for \findsegments}

\begin{algorithm}[ht!]

\caption{Algorithm $\findsegments(\mathcal{P}, \Lambda)$ for finding optimal segmentation for fixed groups $\mathcal{P}$ and parameters $\Lambda$}

$p(v) \define $ group index of $v$\;
$f[e, h] \define 0$ for each $e \in E$ and $h = 1,\ldots,H$\;
$t_{\mathit{min}} \define \min \set{t \mid (u, v, t) \in E}$\;

\ForEach {$h = 1,\ldots,H$} {
	$\alpha \define \sum_{i \leq j} \abs{P_i \times P_j}\lambda_{ijh}$\;
	$\beta \define 0$\;
	\ForEach {$e = (u, v, t) \in E$ in chrono. order} {
		$i \define p(u)$\;
		$j \define p(v)$\;
		$\beta \define \beta + \log \lambda_{ijh}$\;
		$f[e, h] \define \beta - \alpha (t - t_{\mathit{min}})$\;
	}
}

\ForEach {$e \in E$ and $k = 1,\ldots,K$} {
	$o[e, k] \define r[e, k] \define q[e, k] \define 0$\;
}
\ForEach {$e \in E$} {
	$o[e, 1] \define \max_h f[e, h]$\;
	$r[e, 1] \define \arg\max_h f[e, h]$\;
}

\ForEach {$k = 2, \ldots, K$} {
	$z[e, h] \define 0$ for each $e \in E$ and $h = 1,\ldots,H$\;
	$x(s, e; h) \define o[s, k - 1] + f[e, h] - f[s, h]$\;
	\ForEach {$h = 1, \ldots, H$} {
		$z[e, h] \define \arg \max_s x(s, e; h)$ for each $e \in E$ use SMAWK\;
	}
	$r[e, k] \define \arg\max_h x(z[e, h], e; h)$\;
	$o[e, k] \define \max_h x(z[e, h], e; h)$\;
	$q[e, k] \define z[e, r[e, k]]$\;
}

$e \define $ last edge in $E$\;

\ForEach {$k = K, \ldots, 2$} {
	$s \define q[e, k]$\;
	$T_k \define (t(s), t(e)]$\;
	$g(k) \define r[e, k]$\;
	$e \define s$\;
}
$T_1 \define [t_{min}, t(e)]$\;
$g(1) \define r[e, 1]$\;

\Return $(T_1, \ldots, T_K)$, $g$\;

\end{algorithm}

\end{document}